\documentclass[suppldata]{interact}

\usepackage{epstopdf}
\usepackage[caption=false]{subfig}

\usepackage[numbers,sort&compress]{natbib}
\bibpunct[, ]{[}{]}{,}{n}{,}{,}
\renewcommand\bibfont{\fontsize{10}{12}\selectfont}
\makeatletter
\def\NAT@def@citea{\def\@citea{\NAT@separator}}
\makeatother

\theoremstyle{plain}
\newtheorem{theorem}{Theorem}[section]
\newtheorem{lemma}[theorem]{Lemma}
\newtheorem{corollary}[theorem]{Corollary}
\newtheorem{proposition}[theorem]{Proposition}

\theoremstyle{definition}
\newtheorem{definition}[theorem]{Definition}

\theoremstyle{remark}
\newtheorem{remark}{Remark}

\usepackage{amssymb}
\usepackage{verbatim}
\usepackage{blkarray}
\usepackage{multirow}
\usepackage{graphicx}
\usepackage{afterpage}
\usepackage[hidelinks]{hyperref}
\usepackage{framed}
\usepackage{amsmath}
\usepackage{systeme}

\newcommand{\ket}[1]{|#1\rangle}

\newcommand{\vp}{\varphi}

\newcommand{\bs}[1]{\{0,1\}^{#1}}
\newcommand{\bsp}{\{0,1\}}
\newcommand{\bff}[2]{f:\bs{#1}\to\bs{#2}}

\DeclareMathOperator{\GPK}{GPK}

\begin{document}

\title{A generalisation of the Phase Kick-Back}

\author{
\name{Joaqu\'in Ossorio--Castillo\textsuperscript{{$\ast$}}\thanks{\textsuperscript{$\ast$} ORCiD: 0000-0001-8592-2263}, Ulises Pastor--D\'iaz\textsuperscript{{$\dagger$}}\thanks{\textsuperscript{$\dagger$} CONTACT: Ulises Pastor--D\'iaz. Email: ulisespastordiaz@gmail.com, ORCiD: 0000-0002-0309-7173} and Jos\'e~M. Tornero\textsuperscript{{$\ddagger$}}\thanks{\textsuperscript{$\ddagger$} ORCiD: 0000-0001-5898-1049}}
\affil{\textsuperscript{$\ast,\dagger,\ddagger$} Departamento de Álgebra, Facultad de Matemáticas, Universidad de Sevilla. Avda. Reina Mercedes s/n, 41012 Sevilla (Spain).}
}

\maketitle

\begin{abstract}
In this paper, we present a generalisation of the Phase Kick-Back technique, which is central to some of the classical algorithms in quantum computing, such as the Deutsch--Jozsa algorithm, Simon's algorithm or Grover's algorithm.

We will begin by recalling the Phase Kick-Back technique to then introduce the new generalised version and analyse it. After that, we will present a new generalised version of the Deutsch--Jozsa problem and it will be solved using the previously defined technique.

Finally, we will present a generalised version of the Bernstein-Vazirani problem and solve it using this technique to better understand its inner workings.
\end{abstract}


\begin{amscode}
68Q12 (primary); 68Q09, 81P68 (secondary).
\end{amscode}

\begin{section}{Introduction: Phase Kick-Back and notation}

Let us begin by introducing the notation we will use, which will be that of \cite{kaye,niel}. These two books, along with \cite{via,oyt} can be consulted for more context on the topic of quantum computing.

\begin{remark}

First of all, we will call the elements $\mathbf{x}\in\bs{n}$ binary strings and note them in bold, underlining their structure as vectors in the space $\mathbb{F}_2^n$. 

Let $\mathbf{y}, \mathbf{z}\in\{0,1\}^n$ be two strings, written
$$ \mathbf{y} = y_{n-1} \ldots y_1 y_0, \quad \quad \mathbf{z} = z_{n-1} \ldots z_1 z_0, $$
and let $\oplus$ denote the exclusive or addition (which is addition modulo $2$). We define the exclusive or operation for strings as the exclusive or bitwise, that is,
$$ \mathbf{y} \oplus \mathbf{z} = \left( y_{n-1} \oplus z_{n-1} \right) \ldots \left( y_1 \oplus z_1 \right) \left( y_0 \oplus z_0 \right), $$
and we will denote the pairing in $\{0,1\}^n$ (not a scalar product, though) by 
$$ \mathbf{y} \cdot \mathbf{z} = \left( y_0 \cdot z_0  \right) \oplus \ldots \oplus \left( y_{n-1}\cdot z_{n-1} \right). $$

Note that, as the xor operation is performed bitwise, we have
$$ \mathbf{x} \cdot ( \mathbf{y} \oplus \mathbf{z} ) = (\mathbf{x} \cdot \mathbf{y}) \oplus (\mathbf{x} \cdot \mathbf{z}). $$

We will also write $\mathbf{0}$ to refer to the zero $n$-string $\mathbf{0} = 00\cdots 0$.

To represent quantum states we will use the Bra-Ket or Dirac notation, where given a binary string $\mathbf{x}\in\bs{n}$ of length $n$ we represent the $n$-dimensional qubit state of the computational basis corresponding to $\mathbf{x}$ by $\ket{\mathbf{x}}_n$. For one-dimensional qubit systems, we will often simply write the ket $\ket{\mathbf{x}}$ without the subindex. If we have more than one qubit system, we will write the number of qubits of each register separated by commas. For example, in $\ket{\mathbf{x}}_{n,m,r}$ we would have three registers of $n$, $m$ and $r$ qubits respectively.

Let $R$ be an $m\times n$ Boolean matrix---i.e., a matrix whose components are either $0$s or $1$s---and let $\mathbf{r}_i$ be the binary string determined by the $i$-th file of $R$, we will define the result of the operation $R\cdot \mathbf{x}$ as the string whose $i$-th component is $\mathbf{r}_i\cdot \mathbf{x}$ (that is, the usual matrix-vector operation).

\end{remark}

We will say that a Boolean function is a function $\bff{n}{m}$. It is well known---consult \cite{via} for more information---that given a Boolean function one can construct the quantum gate $\mathbf{U}_f$ whose effect is the following:
$$
\mathbf{U}_f\Big(\ket{\mathbf{x}}_n\otimes\ket{\mathbf{y}}_m\Big) = \ket{\mathbf{x}}_n\otimes\ket{\mathbf{y}\oplus f(\mathbf{x})}_m.
$$

Let us now review the Phase Kick-Back. To do so we must recall the \textit{Hadamard basis}:
$$
\ket{+} = \frac{\ket{0}+\ket{1}}{\sqrt{2}}, \quad \ket{-} = \frac{\ket{0}-\ket{1}}{\sqrt{2}}.
$$

\begin{lemma}
Let $\bff{n}{}$ be a Boolean function, and let $\mathbf{U}_f$ be the quantum gate that computes it. Then, in the $n+1$ qubit system, vectors of the form $\ket{\mathbf{x}}_n\otimes\ket{-}$ are eigenvectors with eigenvalue $(-1)^{f(\mathbf{x})}$ for every $\mathbf{x}\in\bs{n}.$

\end{lemma}

\begin{proof}

To prove this result we must expand the following expression:
$$
\mathbf{U}_f\left(\ket{\mathbf{x}}_n\otimes\ket{-}\right) = \ket{\mathbf{x}}_n\otimes\left(\frac{\ket{f(\mathbf{x})}-\ket{f(\mathbf{x})\oplus 1}}{\sqrt{2}}\right).
$$

If $f(\mathbf{x}) = 0$, then the state does not change and we have:
$$
\ket{\mathbf{x}}_n\otimes\left(\frac{\ket{0}-\ket{1}}{\sqrt{2}}\right) = \ket{\mathbf{x}}_n\otimes\ket{-}.
$$

If $f(\mathbf{x}) = 1$, then:
$$
\ket{\mathbf{x}}_n\otimes\left(\frac{\ket{1}-\ket{0}}{\sqrt{2}}\right) = (-1)\left(\ket{\mathbf{x}_n}\otimes\ket{-}\right).
$$
\end{proof}

The \textit{Phase Kick-Back} technique is almost always used to mark the amplitudes of the states of the computational basis whose image through $f$ is $1$. In that sense, we would have
$$
\Big(\mathbf{H}_n\ket{\mathbf{0}}_n\Big)\otimes\ket{-} = \left(\frac{1}{\sqrt{2^n}}\sum_{\mathbf{x}\in\bs{n}} \ket{\mathbf{x}}_n\right)\otimes\ket{-},
$$
where $\mathbf{H}_n$ is the Hadamard matrix of dimension $n$, which can be defined as:
$$
\mathbf{H}_n = \mathbf{H}^{\otimes n}, \, \text{ where } \mathbf{H} = \frac{1}{\sqrt{2}} \left( \begin{array}{rr} 1 & 1 \\ 1 & -1 \end{array} \right),
$$
and whose effect on an element of the computational basis $\mathbf{x}\in\bs{n}$ is the following:
$$
\mathbf{H}_n\ket{\mathbf{x}}_n = \frac{1}{\sqrt{2^n}} \sum_{\mathbf{z}\in\bs{n}} (-1)^{\mathbf{x}\cdot\mathbf{z}} \ket{\mathbf{z}}_n.
$$

This can be easily proven by induction. In particular, when $\mathbf{x} = \mathbf{0}$, we would have:
$$
\mathbf{H}_n\ket{\mathbf{0}}_n = \frac{1}{\sqrt{2^n}} \sum_{\mathbf{z}\in\bs{n}} \ket{\mathbf{z}}_n.
$$

Summarising, we would have a summation over all the states of the computational basis, all of them with the same amplitude in the first $n$-qubit register. The idea of the Phase Kick-Back is to apply $\mathbf{U}_f$ to this state and mark the aforementioned elements with a negative amplitude.
$$
\mathbf{U}_f\left(\frac{1}{\sqrt{2^n}}\sum_{\mathbf{x}\in\bs{n}} \ket{\mathbf{x}}_n\otimes\ket{-}\right) = \left(\frac{1}{\sqrt{2^n}}\sum_{\mathbf{x}\in\bs{n}} (-1)^{f(\mathbf{x})}\ket{\mathbf{x}}_n\right)\otimes\ket{-}.
$$

\begin{remark}{(Deutsch--Jozsa algorithm.)}
Let us recall the Deutsch--Jozsa algorithm as an example of usage of this technique, which was presented in \cite{dyj} but can be reviewed in any of the manuals of quantum computing already presented.

Deutsch's problem is defined as follows. We have an unknown Boolean function $\bff{n}{}$ which can either be constant---that is, $f(\mathbf{x})$ is the same for every $\mathbf{x}\in\bs{n}$---or balanced, which means that for half of the values $\mathbf{x}\in\bs{n}$ we have $f(\mathbf{x}) = 0$ and for the other half $f(\mathbf{x}) = 1$.

Deutsch's problem consists in determining whether $f$ is constant or balanced using the function as a black box. In terms of time efficiency, this problem is actually quite hard to solve in the classical deterministic situation, as it would take $2^{n-1}+1$ evaluations of $f$ to solve in the worst scenario, but the Deutsch--Jozsa algorithm solves it with certainty with a single call to $\mathbf{U}_f$. 

Let us review the role of the Phase Kick-Back in this algorithm. 

Starting with the state $\ket{\vp_0}_{n,1} = \ket{\mathbf{0}}_n\otimes\ket{0}$, we apply the Pauli $\mathbf{X}$ gate to the second register, obtaining $\ket{\vp_1}_{n,1} = \ket{\mathbf{0}}_n\otimes\ket{1}$. Then, after applying Hadamard gates we get the state already introduced given by $\ket{\vp_2}_{n,1}= \Big(\mathbf{H}_n\ket{\mathbf{0}}_n\Big)\otimes\ket{-}$, and apply $\mathbf{U}_f$. Thus, we make use of the Phase Kick-Back technique:
$$
\ket{\vp_3}_{n,1} = \mathbf{U}_f\ket{\vp_2}_{n,1} = \mathbf{U}_f\left(\frac{1}{\sqrt{2^n}}\sum_{\mathbf{x}\in\bs{n}} \ket{\mathbf{x}}_n\otimes\ket{-}\right),
$$
and we end up getting:
$$
\left(\frac{1}{\sqrt{2^n}}\sum_{\mathbf{x}\in\bs{n}} (-1)^{f(\mathbf{x})}\ket{\mathbf{x}}_n\right)\otimes\ket{-}.
$$

Now, if $f$ were constant, then we would actually have gotten $\Big(\mathbf{H}_n\ket{\mathbf{0}}_n\Big)\otimes\ket{-}$ again, so after applying $\mathbf{H}_n$ to the first $n$-qubit register we would get $\ket{\mathbf{0}}_n$. Let us do the calculations.
$$
\ket{\vp_4}_n = \mathbf{H}_n\left(\frac{1}{\sqrt{2^n}}\sum_{\mathbf{x}\in\bs{n}} (-1)^{f(\mathbf{x})}\ket{\mathbf{x}}_n\right) = \frac{1}{\sqrt{2^n}}\sum_{\mathbf{x}\in\bs{n}} (-1)^{f(\mathbf{x})}\mathbf{H}_n\ket{\mathbf{x}}_n.
$$

And using that:
$$
\mathbf{H}_n\ket{\mathbf{x}}_n = \frac{1}{\sqrt{2^n}}\sum_{\mathbf{z}\in\bs{n}} (-1)^{\mathbf{x}\cdot\mathbf{z}}\ket{\mathbf{z}}_n
$$

We can expand the previous expression:
$$
\ket{\vp_4}_n = \frac{1}{\sqrt{2^n}}\sum_{\mathbf{x}\in\bs{n}} (-1)^{f(\mathbf{x})}\left(\frac{1}{\sqrt{2^n}}\sum_{\mathbf{z}\in\bs{n}} (-1)^{\mathbf{x}\cdot\mathbf{z}}\ket{\mathbf{z}}_n \right)
$$
$$
= \frac{1}{2^n} \sum_{\mathbf{z}\in\bs{n}}\left[\sum_{\mathbf{x}\in\bs{n}} (-1)^{\mathbf{x}\cdot\mathbf{z}\oplus f(\mathbf{x})}\right]\ket{\mathbf{z}}_n.
$$

Here, the amplitude of a state of the computational basis $\ket{\mathbf{z}}_n$ is given by:
$$
\frac{1}{2^n}\sum_{\mathbf{x}\in\bs{n}} (-1)^{\mathbf{x}\cdot\mathbf{z}\oplus f(\mathbf{x})}.
$$

In particular, if $\ket{\mathbf{z}}_n = \ket{\mathbf{0}}_n$, then we would get the amplitude:
$$
\frac{1}{2^n}\sum_{\mathbf{x}\in\bs{n}} (-1)^{f(\mathbf{x})}.
$$

That is, $0$ when the function is balanced and $1$ when it is constant. Thus, we end up with the state $\ket{\mathbf{0}}_n$ if the function is constant and a combination of the rest of states in the computational basis if the function is balanced. If we finish by measuring this register, we will get $\mathbf{0}$ if the function is constant and any other result if instead it is balanced.
\end{remark}

\end{section}

\begin{section}{Generalised Phase Kick-Back}

Let us now present a generalisation of the Phase Kick-Back idea. This generalisation will consist on the expansion of the technique to general Boolean functions $\bff{n}{m}$, where the target qubit---the one in the second register of the Deutsch--Jozsa algorithm---becomes a register of $m$ qubits.

During this generalisation we will take $\mathbf{U}_f$ as presented before and we will notate the states given by $\mathbf{H}_n\ket{\mathbf{y}}_n$ as $\ket{\gamma_{\mathbf{y}}}_n$, where $\ket{\mathbf{y}}_n$ are the elements of the computational basis.

Let us begin by presenting an analogous version to that of Lemma $1.1$, which will constitute the core idea of this technique.

\begin{lemma}
Let $\ket{\gamma_{\mathbf{y}}}_m = \mathbf{H}_m\ket{\mathbf{y}}_m$ with $\mathbf{y} \in\{0,1\}^m$. Then, for each $\mathbf{x} \in \{0,1\}^n$, the vector $\ket{\mathbf{x}}_n\otimes\ket{\gamma_{\mathbf{y}}}_m$ is an eigenvector of $\mathbf{U}_f$ with eigenvalue $(-1)^{\mathbf{y}\cdot f(\mathbf{x})}.$ 
\end{lemma}

\begin{proof}

We know that
$$
\ket{\gamma_{\mathbf{y}}}_m = \frac{1}{\sqrt{2^m}} \sum_{\mathbf{z}\in\{0,1\}^m} (-1)^{\mathbf{y}\cdot \mathbf{z}}\ket{\mathbf{z}}_m.
$$

If we now apply $\mathbf{U}_f$ to $\ket{\mathbf{x}}_n\otimes\ket{\gamma_{\mathbf{y}}}_m$, we get the following:
$$
\mathbf{U}_f\Big( \ket{\mathbf{x}}_n\otimes\ket{\gamma_{\mathbf{y}}}_m\Big) = \ket{\mathbf{x}}_n\otimes\left(\frac{1}{\sqrt{2^m}} \sum_{\mathbf{z}\in\{0,1\}^m} (-1)^{\mathbf{y}\cdot \mathbf{z}}\ket{\mathbf{z}\oplus f(\mathbf{x})}_m \right).
$$

And doing some manipulation, this expression becomes:
$$
= (-1)^{\mathbf{y}\cdot f(\mathbf{x})}\ket{\mathbf{x}}_n\otimes\left(\frac{1}{\sqrt{2^m}} \sum_{\mathbf{z}\in\{0,1\}^m} (-1)^{\mathbf{y}\cdot \mathbf{z}\oplus \mathbf{y}\cdot f(\mathbf{x})}\ket{\mathbf{z}\oplus f(\mathbf{x})}_m \right)
$$
$$
= (-1)^{\mathbf{y}\cdot f(\mathbf{x})}\ket{\mathbf{x}}_n\otimes\left(\frac{1}{\sqrt{2^m}} \sum_{\mathbf{z}\in\{0,1\}^m} (-1)^{\mathbf{y}\cdot (\mathbf{z}\oplus f(\mathbf{x}))}\ket{\mathbf{z}\oplus f(\mathbf{x})}_m \right)
$$
$$
= (-1)^{\mathbf{y}\cdot f(x)}\ket{\mathbf{\mathbf{x}}}\otimes\ket{\gamma_{\mathbf{y}}}_m,
$$
as for a fixed $f(\mathbf{x})$, $\ket{\mathbf{z}\oplus f(\mathbf{x})}_m$ runs through all of $\{0,1\}^m$ just as $\mathbf{z}$ does.
\end{proof}

As we can see, it is a completely analogous idea to the previous one, with the difference that we can now choose a {\em marker}, $\mathbf{y}\in\bs{n}$, which will work as a fixed reference and multiply each $f(\mathbf{x})$. Let us take a look at the inner work of this idea using an example.

We will consider the Boolean function $f:\{0,1\}^3 \to \{0,1\}^2$ which will eliminate the last bit, that is, $f(xyz) = xy$ where $x,y,z\in\{0,1\}$ and $xyz$ stands for the concatenation of bits.

To use our new tool, we will need a $5$-qubit system divided into a $3$-qubit register and a $2$-qubit register, both of them starting on $\ket{\mathbf{0}}$:
$$
\ket{\varphi_0}_5 = \ket{\mathbf{0}}_3\otimes\ket{\mathbf{0}}_2.
$$

We will begin by choosing a marker, i.e., the $\mathbf{y}\in\{0,1\}^2$ that will encode the information we want to look for in $f$. In this case, we will take $\mathbf{y} = 01$, that is, we will mark those values whose image through $f$ is $01$ or $11$. To do so, we will begin by preparing the second register to $\mathbf{y}$, which is easily achieved by applying the Pauli $\mathbf{X}$ gate on the last qubit.

$$\ket{\varphi_1}_5 = \left(\mathbf{I}^{\otimes 4}\otimes \mathbf{X}\right)\ket{\varphi_0}_5 = \ket{\mathbf{0}}_3\otimes\ket{01}_2.$$

Once we have prepared our basic state, we will apply Hadamard gates to all qubits to obtain a superposition state.
$$
\ket{\varphi_2}_5 = \mathbf{H}_5\ket{\varphi_1}_5 = \left(\frac{1}{\sqrt{8}}\sum_{\mathbf{x}\in\{0,1\}^3}\ket{\mathbf{x}}_3\right)\otimes\ket{\gamma_{01}}_2 =\frac{1}{\sqrt{8}}\sum_{\mathbf{x}\in\{0,1\}^3}\Big(\ket{\mathbf{x}}_3\otimes\ket{\gamma_{01}}_2\Big).
$$

Let us remark now that each state of the aforementioned superposition satisfies the conditions of Lemma $2.1$, and thus if we apply the $\mathbf{U}_f$ gate we will mark the states of the superposition depending on their image.
$$
\ket{\varphi_3}_5 = \mathbf{U}_f\ket{\varphi_2}_5 = \frac{1}{\sqrt{8}}\sum_{\mathbf{x}\in\{0,1\}^3}(-1)^{f(\mathbf{x})\cdot 01}\Big(\ket{\mathbf{x}}_3\otimes\ket{\gamma_{01}}_2\Big).
$$

\begin{remark}
Another way of looking at this Generalised Phase Kick-Back idea is to write the state $\ket{\gamma_{\mathbf{y}}}$ as a tensor product of $\ket{+}$ and $\ket{-}$ states. As an example, in the instance we are dealing with we have:
$$
\ket{\gamma_{01}}_2 = \ket{+}\otimes\ket{-}.
$$

In general, for a given $\mathbf{y}\in\{0,1\}^m$ we would have a $\ket{+}$ in the $i$-th position if the $i$-th bit of $\mathbf{y}$ is $0$ and $\ket{-}$ if it is $1$. In that sense, we could look at this Generalised Phase Kick-Back as a cascade of Phase Kick-Backs in those positions of $\mathbf{y}$ in which there is a $1$. 
\end{remark}

Let us now focus our attention on the first $3$-qubit register and use that $f(\mathbf{x})\cdot 01 = \mathbf{x} \cdot 010$:
$$
\ket{\varphi_4}_3 = \frac{1}{\sqrt{8}}\sum_{\mathbf{x}\in\{0,1\}^3}(-1)^{f(\mathbf{x})\cdot 01}\ket{\mathbf{x}}_3 = \frac{1}{\sqrt{8}}\sum_{\mathbf{x}\in\{0,1\}^3}(-1)^{\mathbf{x}\cdot 010}\ket{\mathbf{x}}_3.
$$

And finally, if we apply Hadamard gates to this $3$-qubit system, we will get the state $\ket{010}_3$.
$$
\ket{\varphi_5}_3 = \mathbf{H}_3\ket{\varphi_4}_3 = \ket{010}_3.
$$

It is not clear now how this idea is helpful, as the final result is directly determined by the initial $\mathbf{y}$ we chose, and if we had fixed $\mathbf{y} = 10$, then the final result would have been $100$. However, suppose now that we do not know which of the bits $f$ eliminates, and we want to determine which one it is. We only have three possibilities, and we could easily check with one classical call to $f$ which of the bits is eliminated---simply compute $f(010)$---but it is interesting to do it by using our new tool.

\begin{lemma}

Let $f:\{0,1\}^n\to \{0,1\}^{n-1}$ be a Boolean function that eliminates one bit, then we can use the algorithm above to determine which bit is eliminated.

\end{lemma} 

\begin{proof}
To do so, we just apply the generalised version of the algorithm mentioned above $n-1$ times, using each time one of the vectors of the canonical basis of $\{0,1\}^{n-1}$ as an $\mathbb{F}_2$ vector space. If we denote by $\mathbf{e}_i$ the string of bits whose only $1$ is in the $i$-th position (starting by $0$)---i.e., the $i$-th element of the canonical basis---then each of the $n-1$ iterations of the algorithm would go as follows:
$$
\ket{\varphi_0}_{n,n-1} = \ket{\mathbf{0}}_n\otimes\ket{\mathbf{0}}_{n-1}.
$$

First, we obtain $\mathbf{e}_i$ in the second register by applying the $\mathbf{X}$ gate wherever we need:
$$
\ket{\varphi_1}_{n,n-1} = \ket{\mathbf{0}}_n\otimes\ket{\mathbf{e}_i}_{n-1}.
$$

Second, we apply Hadamard gates:
$$
\ket{\varphi_2}_{n,n-1} = \left(\mathbf{H}_{2n-1}\right)\ket{\varphi_1}_{2n-1}.
$$

Then, we use the $\GPK$ (Generalised Phase Kick-Back):
$$
\ket{\varphi_3}_{n,n-1} = \mathbf{U}_f\ket{\varphi_2}_{2n-1}.
$$

And finally, we apply Hadamard gates to the first register and measure:
$$
\ket{\varphi_4}_{n,n-1} = (\mathbf{H}_n\otimes\mathbf{I}^{\otimes (n-1)})\ket{\varphi_3}_{n,n-1}.
$$

After we have done so with all $n-1$ possible $\mathbf{e}_i$, we will have obtained $n-1$ of the $n$ vectors of the canonical basis of $\mathbb{F}_2^n$, and the one left indicates which of the bits is eliminated.
\end{proof}

This, of course, does not give us an improvement of any sort over the classical case---it is actually the opposite, as we could have just computed the image of $\mathbf{x} = 101010\ldots$ and checked for repeated characters---but it illustrates the inner workings of the technique.

Some other examples such as this could be constructed. Another one is the problem of, given an $f$ that switches one unknown bit, finding out which one. However, we will now focus on a problem in which this idea allows for an improvement over the classical situation.

\end{section}

\begin{section}{The generalised Deutsch--Jozsa problem}

An easy follow-up to the previous section would be to solve a generalised version of the Deutsch--Jozsa problem using this technique. Let us begin by presenting the promised problem.

\begin{definition}{(Generalised Deutsch--Jozsa problem.)}

 We say that a Boolean function is balanced if half of the input values output one string and the other half output another. 
 
 Given then a Boolean function $f:\{0,1\}^n\to \{0,1\}^m$ that can either be constant or balanced, we will denote by Generalised Deutsch--Jozsa problem the one of finding out in which of the cases are we.

\end{definition}

The Deutsch--Jozsa problem is clearly one instance of this general problem where $m = 1$, and thus we will show how we can solve this problem by using an algorithm inspired by that of Deutsch and Jozsa. Let us begin by taking a moment to think about the complexity we are dealing with.

\begin{remark}

It is clear that if we want to solve this problem using classical deterministic methods, we will need something of the order of $\mathcal{O}(2^{n-1})$ applications of $f$. We will see how we can improve this with a quantum algorithm to an order of $\mathcal{O}(m)$ calls to $f$. Note also that this includes the already known case where $m=1$.

\end{remark}

To begin taking a look at this problem, let us limit ourselves to the instance where constant means that $f(\mathbf{x}) = \mathbf{0}$ for every $\mathbf{x}\in\{0,1\}^n$ and balanced means that half of the values are $\mathbf{0}$ and the other half a fixed string different from $\mathbf{0}$.

Let us first expose our algorithm and then worry about the analysis. Given $\mathbf{e}_i = 0^{(m-1)-i} \ 1 \ 0^{i}$, where $i=0,\ldots,m-1$, we will repeat the following algorithm for each $\mathbf{e}_i$, but it could actually be done for any binary string $\mathbf{y}\in\{0,1\}^m$.

\vspace{5mm}

$\mathbb{STEP}$ $1$

$\ket{\varphi_0}_{n,m} = \ket{\mathbf{0}}_n\otimes\ket{\mathbf{0}}_m.$

\vspace{5mm}

We begin with two registers of $n$ and $m$ qubits, both at the state $\ket{\mathbf{0}}$.

\vspace{5mm}

$\mathbb{STEP}$ $2$

$\ket{\varphi_1}_{n,m} = \left(\mathbf{I}^{\otimes n}\otimes\mathbf{I}^{\otimes (m-1-i)}\otimes\mathbf{X}\otimes\mathbf{I}^{\otimes i}\right)\ket{\varphi_0}_{n,m}.$

\vspace{5mm}

We apply the $\mathbf{X}$ gate to achieve the desired $\ket{\mathbf{e}_i}$ state in the second register. If we want any other binary string $\mathbf{y}$ to act as a marker, we should apply the corresponding $\mathbf{X}$ gates in the necessary positions.

\vspace{5mm}

$\mathbb{STEP}$ $3$

$\ket{\varphi_2}_{n,m} = \mathbf{H}_{n+m}\ket{\varphi_1}_{n,m}.$

\vspace{5mm}

We apply Hadamard gates to obtain the desired superposition in the first register and $\ket{\gamma_{\mathbf{e}_i}}$ in the second.

\vspace{5mm}

$\mathbb{STEP}$ $4$

$\ket{\varphi_3}_{n,m} = \mathbf{U}_f\ket{\varphi_2}_{n,m}.$

\vspace{5mm}

We apply $\mathbf{U}_f$ to use the $\GPK$ technique.

\vspace{5mm}

$\mathbb{STEP}$ $5$

$\ket{\varphi_4}_{n,m} = \left(\mathbf{H}^{\otimes n}\otimes \mathbf{I}^{\otimes m}\right)\ket{\varphi_3}_{n,m}.$

\vspace{5mm}

At this point, the second register might be discarded and we apply Hadamard gates to the first one.

\vspace{5mm}

$\mathbb{STEP}$ $6$

We measure the first register and name the result $\delta_i$.

\vspace{5mm}

If after repeating these steps for each $i$ we obtain only $\delta_i = \mathbf{0}$ strings, then the function is constant; otherwise it is balanced.

\begin{definition}{(Generalised Phase Kick-Back algorithm.)}
The only variable in the algorithm is the choice of the marker $\mathbf{y}$ used for the Phase Kick-Back.  We will refer to this algorithm as \textit{GPK algorithm for $\mathbf{y}$} or $\GPK(\mathbf{y})$. From now on, the notation regarding this algorithm will be the same as before. 
\end{definition}

\begin{theorem}{(Correctness of the algorithm.)}
The aforementioned algorithm correctly determines whether a function is constant or balanced in the case where the image set of $f$ includes $\mathbf{0}$.
\end{theorem}

\begin{proof}

Given $i=0,\ldots,m-1$, let us keep track of the states step by step:

As we are applying the $\mathbf{X}$ gate on the $i$-th qubit of the second register (counting from $0$), then
$$
\ket{\varphi_1}_{n,m} = \ket{\mathbf{0}}_n\otimes\ket{\mathbf{e}_i}_m,
$$

Next,
$$
\ket{\varphi_2}_{n,m} = \frac{1}{\sqrt{N}} \sum_{\mathbf{x}\in\{0,1\}^n} \ket{\mathbf{x}}_n\otimes\ket{\gamma_{\mathbf{e}_i}}_m,
$$
just by the definition of $\ket{\gamma_{\mathbf{e}_i}}$ and the known effect of Hadamard gates on the $\ket{\mathbf{0}}$ state. Finally, we obtain
$$
\ket{\varphi_3}_{n,m} = \left(\frac{1}{\sqrt{N}} \sum_{\mathbf{x}\in\{0,1\}^n} (-1)^{f(\mathbf{x})\cdot \mathbf{e}_i}\ket{\mathbf{x}}_n\right)\otimes\ket{\gamma_{\mathbf{e}_i}}_m,
$$
where $N=2^n$ by applying Lemma $2.1$.

If we focus now only on the first register, we will have the following state:
$$
\frac{1}{\sqrt{N}} \sum_{\mathbf{x}\in\{0,1\}^n} (-1)^{f(\mathbf{x})\cdot \mathbf{e}_i}\ket{\mathbf{x}}_n.
$$

Then, after applying the Hadamard gates, we will have:
$$
\mathbf{H}_n\frac{1}{\sqrt{N}} \sum_{\mathbf{x}\in\{0,1\}^n} (-1)^{f(\mathbf{x})\cdot \mathbf{e}_i}\ket{\mathbf{x}}_n =  \frac{1}{\sqrt{N}} \sum_{\mathbf{x}\in\{0,1\}^n} (-1)^{f(\mathbf{x})\cdot \mathbf{e}_i}\mathbf{H}_n\ket{\mathbf{x}}_n
$$
$$
= \frac{1}{\sqrt{N}} \sum_{\mathbf{x}\in\{0,1\}^n} (-1)^{f(\mathbf{x})\cdot \mathbf{e}_i}\left(\frac{1}{\sqrt{N}}\sum_{\mathbf{z}\in\{0,1\}^n} (-1)^{\mathbf{x}\cdot \mathbf{z}}\ket{\mathbf{z}}_n\right)
$$
$$
= \frac{1}{N}\sum_{\mathbf{z}\in\{0,1\}^n} \left[\sum_{\mathbf{x}\in\{0,1\}^n} (-1)^{f(\mathbf{x})\cdot \mathbf{e}_i \oplus \mathbf{x}\cdot \mathbf{z}}\right]\ket{\mathbf{z}}_n.
$$

It is easy to check that if the function is constant and equal to $\mathbf{0}$, then regardless of the value of $i$ the amplitude of $\ket{\mathbf{0}}_n$ in the previous superposition is the following:
$$
\frac{1}{N}\sum_{\mathbf{x}\in\{0,1\}^n} (-1)^{f(\mathbf{x})\cdot \mathbf{e}_i} = \frac{1}{N}\sum_{\mathbf{x}\in\{0,1\}^n} (-1)^0 = 1.
$$

Thus, we will always obtain $\delta_i = \mathbf{0}$ no matter which marker we use.

If $f$ is not constant, then when $f(\mathbf{x}) \neq \mathbf{0}$ there must be an $i\in\{0,\ldots,m-1\}$ for which $f(\mathbf{x})\cdot \mathbf{e}_i = 1$. If we take such a $\mathbf{e}_i$, then the amplitude for $\ket{\mathbf{0}}_n$ is:
$$
\frac{1}{N}\sum_{\mathbf{x}\in\{0,1\}^n} (-1)^{\left(f(\mathbf{x})\cdot \mathbf{e}_i\right) \oplus \left(\mathbf{x} \cdot \mathbf{0}\right) } = \frac{1}{N}\sum_{\mathbf{x}\in\{0,1\}^n} (-1)^{f(\mathbf{x})\cdot \mathbf{e}_i} = 0,
$$
because $f(\mathbf{x})$ is balanced and thus half the elements of the sum will be $1$ and the other half $-1$. This implies that we would get a result different from $\mathbf{0}$ for that $i$.
\end{proof}

Note that the choice of the canonical basis is not compulsory and that we could have chosen any other basis of $\mathbb{F}_2^m$ as our markers.

Let us take a moment to prove that the same idea works for the general case of the Generalised Deutsch--Jozsa problem.

\begin{theorem}{(General correctness.)}
The previous algorithm correctly determines whether a function is constant or balanced.
\end{theorem}

\begin{proof}
The only thing left to analyse is the final amplitudes in the general case. To do so, we need to recall that the final state is:
$$
\frac{1}{N}\sum_{\mathbf{z}\in\{0,1\}^n} \left[\sum_{\mathbf{x}\in\{0,1\}^n} (-1)^{\left(f(\mathbf{x})\cdot \mathbf{e}_i\right) \oplus \left(\mathbf{x}\cdot \mathbf{z}\right)}\right]\ket{\mathbf{z}}_n.
$$

If we analyse now the amplitude of $\ket{\mathbf{z}}_n = \ket{\mathbf{0}}_n$, we would be left with:
$$
\frac{1}{N}\sum_{\mathbf{x}\in\{0,1\}^n} (-1)^{\left(f(\mathbf{x})\cdot \mathbf{e}_i\right) \oplus \left(\mathbf{x} \cdot \mathbf{0}\right)} = \frac{1}{N}\sum_{\mathbf{x}\in\{0,1\}^n} (-1)^{f(\mathbf{x})\cdot \mathbf{e}_i}.
$$

If $f(\mathbf{x})$ is constant, then $f(\mathbf{x})\cdot \mathbf{e}_i$ is either always $0$ or always $1$, as $\mathbf{x}$ varies. Whichever the case, the final amplitude will be either $1$ or $-1$ and thus we will always get $\mathbf{0}$ at the end of the algorithm.

On the other hand, if $f(\mathbf{x})$ is balanced with possible values $\mathbf{f}_1, \mathbf{f}_2\in\{0,1\}^m$ such that $\mathbf{f}_1\neq \mathbf{f}_2$, then there is a $i\in\{0,\ldots,m-1\}$ such that $\mathbf{f}_1\cdot \mathbf{e}_i \neq \mathbf{f}_2\cdot \mathbf{e}_i$, and for that $i$ the amplitude of $\mathbf{z} = \mathbf{0}$ would be:
$$
\frac{1}{N}\sum_{\mathbf{x}\in\{0,1\}^n} (-1)^{f(\mathbf{x})\cdot \mathbf{e}_i}.
$$

As the function is balanced between $\mathbf{f}_1$ and $\mathbf{f}_2$, that amplitude is $0$ and thus we would get a result different from $\mathbf{0}$.
\end{proof}

This algorithm not only allows us to distinguish constant and balanced functions, but it also allows us to determine the values of the function. In the balanced situation, it would not be possible to do that efficiently in a deterministic way.

\begin{corollary}

It is possible to determine the possible values of $f$ by applying the aforementioned algorithm and making a classical call to the function.

\end{corollary}

\begin{proof}

Let us begin by the case in which the possible images are $\mathbf{0}$ and $\mathbf{f}_1$. In this situation, the values of $i$ for which we obtain a result different from $\delta_i = \mathbf{0}$ mark the bits of $\mathbf{f}_1$ that are different from $0$, thus determining exactly the value of $\mathbf{f}_1$, so $\mathbf{f}_1 = \boldsymbol\lambda = \lambda_{m-1}\ldots\lambda_1\lambda_0$, where we define $\lambda_i$ as:

$$\lambda_i = \begin{cases} 0 & \text{ if } \delta_i = \mathbf{0} \\ 1 & \text{ otherwise}.
\end{cases}$$

In the general case, if we note the two possible images by $\mathbf{f}_1$ and $\mathbf{f}_2$, the $\boldsymbol\lambda = \lambda_{m-1}\ldots\lambda_1\lambda_0$ string tells us that the Boolean bitwise difference between $\mathbf{f}_1$ and $\mathbf{f}_2$---i.e., $\mathbf{f}_1\oplus \mathbf{f}_2$---. Thus, we would know that $\mathbf{f}_1 = \mathbf{f}_2 \oplus \boldsymbol\lambda$. If we now classically calculate one of the possible images---for instance $f(\mathbf{0})$---we would be able to retrieve both values.
\end{proof}

\begin{remark}

We also have to point out that we have solved the problem by applying the quantum gate $\mathbf{U}_f$ $m$ times, which is an exponential improvement over the deterministic classical situation when $m$ is of linear order with respect to $n$.

\end{remark}

Before we move on to the next generalised problem, let us take a moment to make two important remarks.

\begin{remark}

The first is about a certain pattern that will reappear in the next section, which is that the $\GPK$ algorithm is unable to detect translations. That is, given two Boolean functions $f_1, f_2:\bs{n}\to\bs{m}$ for which there is an $\mathbf{s}\in\bs{n}$ such that $f_1(\mathbf{x}) = f_2(\mathbf{x})\oplus \mathbf{s}$ for every $\mathbf{x}\in\bs{n}$, if we analyse the first register of $\ket{\vp_4}_{n+m}$ for function $f_2$ using $\mathbf{y}\in\bs{m}$ as a marker, we get:
$$
\frac{1}{N}\sum_{\mathbf{z}\in\{0,1\}^n} \left[\sum_{\mathbf{x}\in\{0,1\}^n} (-1)^{f_2(\mathbf{x})\cdot \mathbf{y} \oplus \mathbf{x}\cdot \mathbf{z}}\right]\ket{\mathbf{z}}_n.
$$

And if we now use that $f_2(\mathbf{x}) = f_1(\mathbf{x})\oplus\mathbf{s}$, we get:
$$
\frac{1}{N}\sum_{\mathbf{z}\in\{0,1\}^n} \left[\sum_{\mathbf{x}\in\{0,1\}^n} (-1)^{\left(f_1(\mathbf{x})\oplus \mathbf{s}\right)\cdot \mathbf{y} \oplus \mathbf{x}\cdot \mathbf{z}}\right]\ket{\mathbf{z}}_n
$$
$$
= (-1)^{\mathbf{s}\cdot\mathbf{y}}\frac{1}{N}\sum_{\mathbf{z}\in\{0,1\}^n} \left[\sum_{\mathbf{x}\in\{0,1\}^n} (-1)^{f_1(\mathbf{x})\cdot \mathbf{y} \oplus \mathbf{x}\cdot \mathbf{z}}\right]\ket{\mathbf{z}}_n.
$$

And, as we can observe, we end up getting a quantum state equivalent to the one we would get by applying the $\GPK$ algorithm for the function $f_1$ which does not affect the probabilities of the final result. This is the reason behind the fact that what we get in the general case of the balanced situation in the Generalised Deutsch--Jozsa algorithm is the sum of the two possible values $\boldsymbol\lambda$, and why we must make an extra step to find both values.

\end{remark}

The other thing we want to point out has to do with the choice of markers.

\begin{remark}
In order to solve the generalised Deutsch--Jozsa problem we have computed $m$ applications of the $\GPK$ algorithm with the elements of the computational basis as markers. What we want to show now is that this choice of markers is not compulsory and that any basis of $\bs{m}$ would suffice.

Let $\mathbf{y}_1,\ldots,\mathbf{y}_m\in\bs{m}$ be any such basis, we will compute now the $\GPK$ algorithm for each of these markers. It becomes clear that if $f(\mathbf{x})\cdot\mathbf{y}_i$ is constant for all $\mathbf{x}\in\bs{n}$, then the result of the $i$-th iteration of the algorithm will be $\mathbf{0}$, while if $f(\mathbf{x})\cdot\mathbf{y}_i= 0$ for half of the values and $1$ for the other half, then the result will be any other binary string.

Let $\boldsymbol\lambda = \mathbf{f}_1\oplus\mathbf{f}_2$ be the sum of the two possible values of the function as before---if the function is constant we would have $\boldsymbol\lambda = \mathbf{0}$---then what we end up with is a system of equations:
$$
\{\mathbf{y}_i\cdot\boldsymbol\lambda = \delta_i\mid i = 1,\ldots, m\}.
$$

Where $\boldsymbol\lambda$ is the string of unknowns. This system is always made up of $m$ linearly independent equations, as the $\mathbf{y}_i$ are a basis of $\bs{m}$, so the sole solution will be the desired $\boldsymbol\lambda$.

\end{remark}

\end{section}

\begin{section}{A Bernstein--Vazirani inspired algorithm}

Once again we will put our focus on generalising an already known problem which was studied in \cite{byv}. Let us begin by recalling the Bernstein--Vazirani problem in the one-dimensional situation.

\begin{definition}{(Bernstein--Vazirani problem.)}
Let $f:\{0,1\}^n\to\{0,1\}$ be a function such that there is an $\mathbf{r} \in \{0,1\}^n$ for which $f(\mathbf{x}) = \mathbf{r} \cdot \mathbf{x}$, we want to find the binary string $\mathbf{r}$.

\end{definition}

Before analysing this problem, let us note that the condition stated in the Bernstein--Vazirani problem just asks for $f$ to be linear. This is relevant because in the generalisation of this problem we will consider a linear $\bff{n}{m}$ and ask to exactly determine it.

Regarding the complexity of this problem, we should note that a linear function $\bff{n}{}$ can be determined in $n$ calls to $f$, as we only have to calculate the image through $f$ of the elements of one basis of $\bs{n}$. In particular, we can calculate $f(\mathbf{e}_i)$ for each element in the canonical basis and the $i$-th element of $\mathbf{r}$ would be $r_i = f(\mathbf{e}_i)$.

We will show how we can solve this problem with a quantum algorithm making a single call to $\mathbf{U}_f$. The algorithm we will describe is exactly the same as we used to solve the Deutsch--Jozsa problem.

First, we will have two registers of $n$ and $1$ qubits respectively:
$$
\ket{\vp_{0}}_{n,1} = \ket{\mathbf{0}}_n\otimes\ket{1} 
$$

We can obtain the $\ket{1}$ in the second register by applying the $\mathbf{X}$ gate to the last qubit. Secondly, we will apply Hadamard gates to all the qubits in order to obtain the state:
$$
\ket{\vp_{1}}_{n,1} = \mathbf{H}_n\ket{\vp_0}_{n,1} = \left(\frac{1}{\sqrt{N}} \sum_{\mathbf{x}\in\bs{n}} \ket{\mathbf{x}}_n\right)\otimes\ket{-},
$$
where $N = 2^n$. This state is now ready to use the Phase Kick-Back technique by applying $\mathbf{U}_f$:
$$
\ket{\vp_{2}}_{n,1} = \mathbf{U}_f\ket{\vp_1}_{n,1} = \left(\frac{1}{\sqrt{N}} \sum_{\mathbf{x}\in\bs{n}} (-1)^{f(\mathbf{x})}\ket{\mathbf{x}}_n\right)\otimes\ket{-}.
$$

Using now that $f(\mathbf{x}) = \mathbf{r}\cdot \mathbf{x}$, we arrive at:
$$
\ket{\vp_{2}}_{n,1} = \left(\frac{1}{\sqrt{N}} \sum_{\mathbf{x}\in\bs{n}} (-1)^{\mathbf{r}\cdot\mathbf{x}}\ket{\mathbf{x}}_n\right)\otimes\ket{-}.
$$

Recalling the effect of $\mathbf{H}_n$ on the computational basis, we can easily check that the first register of this state is exactly $\mathbf{H}_n\ket{\mathbf{r}}$, so after applying $\mathbf{H}_n$ to the first register we obtain:
$$
\ket{\vp_{3}}_{n,1} = \left(\mathbf{H}_n\otimes \mathbf{I}\right)\ket{\vp_2}_{n,1} = \ket{\mathbf{r}}_n\otimes\ket{-}.
$$

Then, after measuring the first register we will obtain $\mathbf{r}$. 

Before considering the generalised problem, we will take the liberty to consider a slight modification to the Bernstein--Vazirani problem.

\begin{definition}{(Modified Bernstein--Vazirani problem.)}
Let $\bff{n}{}$ be a Boolean affine function---i.e., a Boolean function such that there are $\mathbf{r}\in\bs{n}$ and $r_0\in\bsp$ for which $f(\mathbf{x}) = r_0\oplus\mathbf{r}\cdot\mathbf{x}$ for all $\mathbf{x}$---then we want to exactly determine said function.

\end{definition}

What we will find out is that this problem can be solved by the previous algorithm with just a final step to determine $r_0$.

\begin{proposition}
The Bernstein--Vazirani algorithm solves the modified Bernstein--Vazirani problem with certainty with a final classical deterministic call to $f$ to determine $r_0$.

\end{proposition}

\begin{proof}

Following the previous exposition of the Bernstein--Vazirani algorithm, the only difference in this situation is that we would end up with the state:
$$
\ket{\vp_2}_{n,1} = \left(\frac{1}{\sqrt{N}} \sum_{\mathbf{x}\in\bs{n}} (-1)^{r_0\oplus\mathbf{r}\cdot\mathbf{x}}\ket{\mathbf{x}}_n\right)\otimes\ket{-} 
$$
$$
= (-1)^{r_0}\left(\frac{1}{\sqrt{N}} \sum_{\mathbf{x}\in\bs{n}} (-1)^{\mathbf{r}\cdot\mathbf{x}}\ket{\mathbf{x}}_n\right)\otimes\ket{-}.
$$

This is equivalent to the state we had in the previous situation, and thus we would end up getting $\mathbf{r}$ after measuring $\ket{\vp_3}_{n,1}$.

To get $r_0$, we must only classically calculate $f(\mathbf{0}) = r_0$.
\end{proof}

Again, we arrive at the same pattern, where the $\GPK$ cannot distinguish a translation in $f$, but only the linear structure it has. 

Let us now use this idea to generalise the Bernstein--Vazirani problem to an arbitrary dimension.

\begin{definition}{(Generalised Bernstein--Vazirani problem.)}
Let $\bff{n}{m}$ be an affine function, i.e., one such that there is an $m\times n$ matrix $R$ and an $\mathbf{r}_0\in\bs{m}$ for which $f(\mathbf{x}) = \mathbf{r}_0\oplus R\cdot\mathbf{x}$. The Generalised Bernstein--Vazirani problem is that of exactly determining $f$.

\end{definition}

\begin{remark}

Let us analyse the classical deterministic complexity of this problem. It is easy to prove that we can exactly determine $R$ by calculating $f(\mathbf{e}_i)$ for each element of the computational basis, as the binary string determined by the $i$-th file of $R$, $\mathbf{r}_i$, will be exactly $f(\mathbf{e}_i)\oplus \mathbf{r}_0$. We can finally calculate $\mathbf{r}_0$ by computing $f(\mathbf{0})$, so the total calls to $f$ will be $n+1$.

It can be seen that with the $\GPK$ we can do this with $m+1$ calls to the function, so in a way we will switch the roles of $\bs{n}$ and $\bs{m}$.

\end{remark}

We will now prove that we can solve the Generalised Bernstein--Vazirani problem by computing $m$ iterations of the $\GPK$ algorithm by each of the elements of the computational basis of $\bs{m}$ and a final classical computation of $f(\mathbf{0})$.

\begin{theorem}{(Correctness of the algorithm.)}
It is possible to exactly determine the matrix $R$ by computing $\GPK(\mathbf{e}_i)$ for each of the elements $\mathbf{e}_i$ of the computational basis of $\bs{m}$.

\end{theorem}

\begin{proof}

We will only prove that the result of the algorithm $\GPK(\mathbf{e}_i)$ is the binary string that determines the $i$-th row of $R$, which is an $\mathbf{r}_i$ such that $f(\mathbf{x})_i = (\mathbf{r}_0)_i\oplus \mathbf{r}_i\cdot \mathbf{x}$.

Let us calculate the amplitude of $\mathbf{r}_i$ in the final state of the $\GPK$ algorithm using $\mathbf{e}_i$ as marker.
$$
\ket{\vp_4}_{n} = \frac{1}{N}\sum_{\mathbf{z}\in\{0,1\}^n} \left[\sum_{\mathbf{x}\in\{0,1\}^n} (-1)^{f(\mathbf{x})\cdot \mathbf{e}_i \oplus \mathbf{x}\cdot \mathbf{z}}\right]\ket{\mathbf{z}}_n.
$$

Therefore, the amplitude of $\mathbf{r}_i$ is:
$$
\frac{1}{N}\sum_{\mathbf{x}\in\{0,1\}^n} (-1)^{f(\mathbf{x})\cdot \mathbf{e}_i \oplus \mathbf{x}\cdot \mathbf{r}_i} = \frac{1}{N}\sum_{\mathbf{x}\in\{0,1\}^n} (-1)^{\left(\mathbf{r}_0\oplus \mathbf{r}_i\cdot \mathbf{x}\right)  \oplus \mathbf{x}\cdot \mathbf{r}_i}.
$$

As $f(\mathbf{x})\cdot \mathbf{e}_i = (\mathbf{r}_0)_i\oplus \mathbf{r}_i\cdot \mathbf{x}$. If we expand now the expression, we get:
$$
\frac{1}{N}\sum_{\mathbf{x}\in\{0,1\}^n} (-1)^{\left(\mathbf{r}_0\oplus \mathbf{r}_i\cdot \mathbf{x}\right)  \oplus \mathbf{x}\cdot \mathbf{r}_i} = (-1)^{\mathbf{r}_0} \frac{1}{N}\sum_{\mathbf{x}\in\{0,1\}^n} (-1)^{\mathbf{x} \cdot\left(\mathbf{r}_i \oplus \mathbf{r}_i\right)} = (-1)^{\mathbf{r}_0},
$$
and we are assured to get $\mathbf{r}_i$.

Once again, $\GPK$ only allows us to determine $R$, but tells us nothing about the translation $\mathbf{r}_0$, which we have to classically determine by computing $f(\mathbf{0})$.
\end{proof}

\begin{remark}
The choice of computing the $\GPK$ algorithm with the elements of the computational basis is actually arbitrary, if we chose to do so with any other basis, we would end up getting the matrix of the linear application in said basis.
\end{remark}

Again, we see that the $\GPK$ shines the most when applied to functions with a certain linear structure. 

\end{section}

\begin{section}{Conclusion and further research}

In this paper we have presented a generalisation of a core technique in quantum computing, which has allowed us to generalise two classical problems in the field. This is important for two reasons.

Firstly, these new generalised problems and theirs solutions allow us to reach further into the understanding of the possibilities and limitations of quantum computing, giving us a deeper look into some of the problems that constitute the foundation of quantum algorithms.

Secondly, this technique may be used for solving some other problems, so we are effectively enlarging the arsenal at our disposal when faced with the uncertainty and unfamiliarity of quantum algorithms.

In future work, we hope to use this technique to target some new and old problems, further exploring the idea of balanced functions in the multidimensional situation.

\end{section}

\section*{Author contributions}

All authors have contributed equally to the work.

\section*{Acknowledgements}

This work was supported by the \emph{Ministerio de Ciencia e Innovaci\'on} under Project PID2020-114613GB-I00 (MCIN/AEI/10.13039/501100011033) and by the \emph{Junta de Andaluc\'ia} and \emph{ERDF} under Project P20-01056.

\section*{Competing interests}

The authors report that there are no competing interests to declare.

\bibliographystyle{nature}
\bibliography{nature}

\end{document}